\RequirePackage{fix-cm}
\documentclass[smallextended,natbib]{svjour3}

\smartqed
\usepackage{breakcites}

\usepackage{float}
\floatstyle{ruled}
\newfloat{algorithm}{tbp}{loa}
\floatname{algorithm}{Algorithm}

\usepackage{amsmath,amsfonts}
\usepackage{tabu,caption}

\usepackage{algorithm}

\usepackage{multirow}
\usepackage{array}

\usepackage[shortcuts]{extdash}

\usepackage{wrapfig}
\usepackage{framed}

\newcommand{\ef}{envy\-/free}
\newcommand{\efness}{envy\-/freeness}
\newcommand{\aef}{averagely\-/\ef{}}
\newcommand{\uef}{unanimously\-/\ef{}}
\newcommand{\mef}{democratically\-/\ef{}}
\newcommand{\aefness}{average\-/\efness{}}
\newcommand{\uefness}{unanimous\-/\efness{}}
\newcommand{\mefness}{democratic\-/\efness{}}

\newcommand{\pr}{proportional}
\newcommand{\prness}{proportionality}
\newcommand{\apr}{averagely\-/\pr{}}
\newcommand{\upr}{unanimously\-/\pr{}}
\newcommand{\mpr}{democratically\-/\pr{}}
\newcommand{\aprness}{average\-/\prness{}}
\newcommand{\uprness}{unanimous\-/\prness{}}
\newcommand{\mprness}{democratic\-/\prness{}}
\newcommand{\unpr}{UnanimousPR}

\newcommand{\exact}{Exact}
\newcommand{\range}[2]{\in\{#1,\dots,#2\}}

\newcommand{\wavg}{W^{\text{avg}}}
\newcommand{\wmin}{W^{\text{min}}}
\newcommand{\wmed}{W^{\text{med}}}

\spnewtheorem*{lemma*}{Lemma}{\bf}{\it}

\def\logbound{\lceil\log_2{k}\rceil}

\def\comp{\textsc{Comp}}

\newtheorem{innercustomthm}{Theorem}
\newenvironment{customthm}[1]
{\renewcommand\theinnercustomthm{#1}\innercustomthm}
{\endinnercustomthm}

\begin{document}

\title{Fair Cake-Cutting among Families%
\thanks{
This research was funded in part by the following institutions: The Doctoral Fellowships of Excellence Program at Bar-Ilan University, the Mordechai and Monique Katz Graduate Fellowship Program, and the Israel Science Foundation grant 1083/13. 
\\ \\
We are grateful to Galya Segal-Halevi, Yonatan Aumann, Avinatan Hassidim, Noga Alon, Christian Klamler, Ulle Endriss, Neill Clift and Sophie Bade for helpful discussions.
We are grateful to the anonymous reviewers of Social Choice and Welfare for their helpful comments, which substantially improved both the contents and presentation of this paper.
\\ \\
This paper started with a discussion in the MathOverflow website at http://mathoverflow.net/questions/203060/fair-cake-cutting-between-groups . We are grateful to the members who participated in the discussion: Pietro Majer, Tony Huynh and Manfred Weis. Other members of the StackExchange network contributed useful answers and ideas: Alex Ravsky, Andrew D. Hwang, BKay, Christian Elsholtz, Daniel Fischer, David K, D.W., Hurkyl, Ittay Weiss, Kittsil, Michael Albanese, Raphael Reitzig, Real, Babou, 
Domotor Palvolgyi 	(domotorp), 
Ian Turton (iant) and ivancho.
}
}
\titlerunning{Fair Cake-Cutting among Families}  

\author{
Erel Segal-Halevi
\and 
Shmuel Nitzan
}

\institute{
Erel Segal Halevi: Department of Computer Science, Ariel University, Ariel 40700, Israel, erelsgl@gmail.com.
\\
Shmuel Nitzan: 
Department of Economics, Bar Ilan University, Ramat Gan 52900, Israel and Hitotsubashi Institute for Advanced Study, Hitotsubashi University. nitzans@biu.ac.il.
}

\date{Received: date / Accepted: date}

\maketitle              

\begin{abstract}
We study the fair division of a continuous resource, such as a land-estate or a time-interval, among pre-specified groups of agents, such as families. 
Each family is given a piece of the resource and this piece is used simultaneously by all family members, while different members may have different value functions.
Three ways to assess the fairness of such a division are examined. 
(a) Average Fairness means that each family's share is fair according to the "family value function", defined as the arithmetic mean of the value functions of the family members. 
(b) Unanimous Fairness means that all members in all families feel that their family's share is fair according to their personal value function.
(c) Democratic Fairness means that in each family, at least a fixed fraction (e.g. a half) of the members feel that their family's share is fair. 
We compare these criteria based on the number of connected components in the resulting division and on their compatibility with Pareto-efficiency.

\keywords{fair division, cake-cutting, public good, club good, no-envy}
\end{abstract}

\section{Introduction}
Fair division of heterogeneous resources among agents with different preferences has been an important issue since Biblical times. Today it is an active area of research in the interface of computer science \citep{Robertson1998CakeCutting,Procaccia2015Cake,Branzei2015Computational,Lindner2016} and economics \citep{Moulin2004Fair,Thomson2011Fair}. Its applications range from politics \citep{Brams1996Fair,Brams2007Mathematics} to multi-agent systems \citep{Chevaleyre2006Issues}.

In most fair division problems, the resource is divided among $n$ individual agents, and the fairness of a division is assessed based on their individual preferences.
A common fairness criterion is \textbf{\prness{}}. It requires that each agent receives a share that is at least as good as $1/n$ of the total endowment, according to the agent's individual preferences.%
\footnote{
The condition of receiving at least $1/n$ of the total endowment was introduced by \citet{Steinhaus1948Problem}.
Economists often call it \emph{fair-share guarantee} \citep{Bogomolnaia2017Competitive}.
Computer scientists often call it  \emph{proportionality}
\citep{Robertson1998CakeCutting}.
}

In practice, however, goods are often owned and used by groups. As an example, consider a land-estate inherited by $k$ families, a river that has to be divided among $k$ states, or the usage-time of a conference room that has to be divided among $k$ meeting groups. The resource (whether land or time) should be divided into $k$ pieces, one piece per group. Each group's share is then used by \emph{all} its members simultaneously. The land-plot allotted to a family is inhabited by the entire family. The share of the river allotted to a state becomes a national park open to all its citizens.
In the time-slot allotted to a group, the conference room is used by all group members.%
\footnote{
In economic terms, the allotted piece becomes a "club good" \citep{buchanan1965economic}. 
}

The happiness of each group member depends on his/her valuation of the entire share of the group. But, in each group there are different members with different valuations. The group's share can be valued by some of its members as at least $1/k$ of the total and by others as less than $1/k$ of the total. 
How, then, should the fairness of a division be assessed? 

The present paper studies this question in the classic setting of \emph{cake-cutting}, introduced by 
\citet{Steinhaus1948Problem}. 
In this setting, there is a measurable space (e.g. an interval or a polygon) called the \emph{cake}. The preferences of each agent are represented by a value-measure on the cake.%
\footnote{
The assumption that agents' preferences can be represented by measures is a strong one. It implies that the agents' value functions are linear --- the sum of values of two disjoint pieces equals the value of their union. 
This implies that the agents have constant marginal utilities \citep{Chambers2005Allocation} --- the marginal utility of a land-plot for an agent does not depend on the other land-plots owned by that agent.
Although this assumption is not always realistic, 
it is often made for practical reasons --- it is much easier to ask people to report linear preferences than to report general preferences. See 
\citet{Bogomolnaia2017Competitive} for a recent discussion of this issue.
Although most papers on the cake-cutting problem assume linearity, there are some notable exceptions; they are surveyed in subsection \ref{sub:nonadditive}.
}

We study three ways to assess the fairness of a division.

First, it is possible to aggregate the valuations in each family to a single \emph{family valuation}. Following the utilitarian tradition \citep{Bentham1789Introduction}, the family-valuation can be defined as the sum or (equivalently) the arithmetic average of the valuations of all family members.
We call a division 
\textbf{average-fair} if it is fair according to these family valuations. In particular, a division is \apr{} if every family receives a share with an average value (averaged over all family members) of at least $1/k$ of its average value of the entire endowment.

By this definition, the family-division problem is easy to solve. 
Since the average of measures is itself a measure, each family can be represented by a single agent, and the problem reduces to fair division among the $k$ representatives. Classic results imply that \apr{} 
allocations exist
\emph{(Section \ref{sec:avg-fairness})}.

Average fairness makes sense only when the numerical values of the agents' valuations are meaningful and they are all measured in the same units, e.g. in dollars (see chapter 3 of \cite{Moulin2004Fair} for some real-life examples of such situations). However, if the valuations represent individual happiness measures that cannot be put on a common scale, then their sum is meaningless, and other fairness criteria should be used.

A second option is to require that all members of every family agree that the division is fair. We call a division 
\textbf{unanimous fair}
if it is fair according to every individual valuation. In particular, a division is \upr{} if every agent values his/her family's share as at least $1/k$ of the total value%
. 
The advantage of this definition is that it does not need to assume that all valuations share a common scale. 
Even though it is a very strong requirement, we prove that \upr{}
allocations exist \emph{(Section \ref{sec:unan-fairness})}.

A disadvantage of unanimous fairness, compared to average fairness, is that unanimously-fair  divisions might be highly fractioned. When an interval is divided, there always exists an \apr{} 
division that is also \emph{connected} --- the share of each family is a single interval (Section \ref{sec:avg-fairness}).
However, 
there might not exist connected \upr{} 
divisions. Moreover, in some cases, the number of intervals in any \upr{} division might be at least $n$ --- the number of individual agents (Section \ref{sec:unan-fairness}). When the number of agents is large, as in the case of dividing land among states, such divisions might be impractical.

In democratic societies, decisions are almost never unanimous. In fact, when the number of citizens is large, it may be impossible to attain unanimity on even the most trivial issue, and decisions are often made by voting. Therefore we suggest a third fairness criterion. Given a fraction $h\in [0,1]$,
we call a division \textbf{$h$-democratic fair} if at least a fraction $h$ of the members in each family consider it fair. Unanimous\-/fairness is equivalent to $1$-democratic fairness.
The case $h={1\over 2}$ is particularly interesting.
${1\over 2}$-democratic fairness can be justified by the following process. After a division is proposed, each family conducts a referendum in which each member approves the division if he/she values the family's share as at least $1/k$ of the total. The division is implemented only if, in every family, a (weak) majority of the population approves it. 

Democratic\-/fairness combines some advantages of unanimous\-/fairness and
average\-/fairness.
It is similar to unanimous\-/fairness in that it does not need to assume that all valuations share a common scale.
It is similar to average\-/fairness in that, when $h\leq {1\over k}$, 
$h$-democratic fairness among $k$ families can be attained with connected pieces. In particular, with $k=2$ families, there always exists an allocation in which each family receives a single connected piece, and at least a weak majority in each family considers the allocation fair. An additional advantage of
democratic fairness in this case is that it can be computed efficiently \emph{(Section \ref{sec:dem-fairness})}.%
\footnote{In contrast, average-fair and unanimous-fair allocations cannot be computed by \emph{any} finite protocol. See Remark \ref{rem:computation} in page \pageref{rem:computation}.
}

Although democratic\-/fairness might leave some citizens unhappy, this may be unavoidable in real-life situations. This is understandable in light of Winston Churchil's dictum: ``democracy is the worst form of government, except all the others that have been tried''.%
\footnote{
A fourth fairness criterion that could be considered is \textbf{individual fairness}. In particular, an allocation is \emph{individually-\pr{}} if the allocation $X=(X_{1},\ldots,X_{k})$ admits a refinement $Y=(Y_{1},\ldots,Y_{n})$, where for each family $F_{j}$, $\cup_{i\in F_{j}}Y_{i}=X_{j}$, such that for each agent $i$, $V_{i}(Y_{i})\geq1/n$. Individually-fair
allocations always exist and can be found by using any classic fair division procedure on the individual agents, disregarding their families. Individual-fairness makes sense if, after the division of the land among the families, each family intends to further divide its share among its members. However, often this is not the case. When an inherited land-estate is divided between two families, the members of each family intend to live and use their entire share together, rather than dividing it among them. Therefore, the happiness of each family member depends on the entire value of his family's share, rather than on the value of a potential private share he would get in a hypothetic sub-division. 
}

While the geometric requirement of having a connected division is practically important, 
an even more important requirement from an economic perspective is Pareto-efficiency. All three variants of \prness{} are compatible with Pareto-efficiency. However, connectivity and Pareto-efficiency are incompatible even without fairness considerations \emph{(Section \ref{sec:efficiency})}.

The \prness{} criterion can be generalized to a situation in which each family has a different entitlement. Suppose that each family $j$ is entitled to a fraction $t_j$ of the resource (where $\sum_{j=1}^k t_j = 1$). An allocation is \upr{} if each member in family $j$ believes that the share of family $j$ is worth at least a fraction $t_j$ of the total. In a similar way it is possible to generalize \aprness{} and \mprness{}. In the simplest setting, the families have equal entitlements, i.e, for each $j\range{1}{k}$: $t_j = 1/k$. Equal entitlements make sense, for example, when $k$ siblings inherit their parents' estate. While an heir will probably like to take his family's preferences into account when selecting a share, each heir is entitled to $1/k$ of the estate regardless of the family size. 

In general, each family may have a different entitlement. The entitlement of a family may depend on its size but may also depend on other factors. For example, consider 
several families who jointly buy a vacation apartment. 
The apartment can host one family at a time, so the families have to divide the year (a time-interval) among them. The entitlement of each family naturally depends on the amount of money it contributed to the purchase, rather than on the family's size.%
\footnote{
See \citet{cseh2018complexity} 
for a recent account of fair division among individual agents with different entitlements.
}
The results presented in \emph{Sections \ref{sec:avg-fairness}--\ref{sec:efficiency}} consider \prness{} both with equal and with different entitlements.

An alternative fairness criterion that is very common in economics is \textbf{\efness{}}. In the context of individual agents, it means that each agent receives a share that is at least as good as the share of any other agent,  
according to the first agent's individual valuation.%
\footnote{
The condition of receiving at least as much as any other agent was introduced by
\citet{Gamow1958Puzzlemath} and 
\citet{Foley1967Resource}.
Economists often call it \emph{no envy} \citep{Bogomolnaia2017Competitive}.
Computer scientists often call it  \emph{envy-freeness} \citep{Robertson1998CakeCutting}.
}
In the context of families, three variants of \efness{} can be defined analogously to the three variants of \prness{} (for families with equal entitlements): \aefness{}, \uefness{} and \mefness{} \emph{(Section \ref{sec:no-envy})}.

From a geometric perspective, these three variants behave similarly to their \prness{} counterparts, that is: 
\begin{itemize}
\item Connected \aef{} allocations always exist;
\item Connected \uef{} allocations are not guaranteed to exist even for two families;
\item Connected ${1\over k}$-\mef{} allocations are guaranteed to exist for $k$ families. In particular, 
connected ${1\over 2}$-\mef{} allocations are guaranteed to exist for two families (but not for three or more families).
\end{itemize}
However, from an economic perspective, \efness{} behaves differently:
\begin{itemize}
\item Pareto-efficient \aef{} allocations always exist;
\item Pareto-efficient \uef{} allocations are guaranteed to exist for two but not for three or more families;
\item Pareto-efficient ${1\over 2}$-\mef{} allocations are guaranteed to exist for two but not for five or more families (we do not know whether they always exist for three or four families).
\end{itemize}

The paper is organized as follows.
Most of the paper focuses on the \emph{\prness{}} criterion. Section \ref{sec:Model-and-Notation} formally presents the model. Sections \ref{sec:avg-fairness}, \ref{sec:unan-fairness} and \ref{sec:dem-fairness} study average, unanimous and democratic \prness{} respectively. We study this criterion both for families with equal entitlements and for families with different entitlements.

Section \ref{sec:efficiency} studies the three variants of \prness{} in combination with Pareto-efficiency. 
Section \ref{sec:no-envy} studies family fairness based on the \efness{} criterion, explaining the differences between the results for \prness{} and for \efness{}. 
Finally, Section \ref{sec:related} compares our work to previous and ongoing related work.

\section{Model and Notation}
\label{sec:Model-and-Notation}
\subsection{Resource and agents}
In the usual cake-cutting setting, there is a resource $C$ (``cake'') that has to be divided. For simplicity it is assumed that $C$ is an interval in $\mathbb{R}$.
A realistic example of such a resource is time:
consider a conference room that can host a single meeting at a time. It is available between 8:00 and 20:00, and this time-interval must be divided among all those who want to use the room.
Another realistic example is the shoreline of a sea or a river: while usually not a straight line, it can be easily mapped to an interval.

There is a set of agents $N = \{1,\ldots,n\}$. Each agent $i\in N$ has a value measure $V_i$, defined on the Borel subsets of $C$. The $V_i$ are assumed to be nonatomic, so that all singular points have a value of 0 to all agents. 
As the term measure implies, the $V_i$ are \emph{additive} --- the value of a union of two disjoint pieces is the sum of the values of the pieces. 
The value measures are normalized such that $\forall i: \, V_i(\emptyset)=0, \, V_i(C)=1$.

\subsection{Families and entitlements}
In our setting, 
there is a set of families $F = \{F_1,\ldots,F_k\}$. We use the term ``family'' to emphasize that the partition of agents to groups is fixed in advance and cannot be modified during the division process.

The number of agents in $F_j$ is denoted $n_j$. Each agent $i\in N$ is a member of exactly one family $F_j\in F$, so $n=\sum_{j=1}^{k}n_j$. 

For each family $F_j$, there is a positive number $t_j$ representing the entitlement of this family. The sum of all entitlements is one: $1 = \sum_{j=1}^k t_j$. 
In the special case in which all families have equal  entitlements, we have for each $j\range{1}{k}$: $t_j = 1/k$. 

\subsection{Allocations and components}
An \emph{allocation} is a vector of $k$ pieces, $X=(X_1,\dots,X_k)$, one piece per family, such that the $X_j$ 
are pairwise-disjoint and $\cup_{j}{X_j} = C$. 

Each piece is a finite union of intervals. We denote by $\comp(X_j)$ the number of connected components (intervals) in the piece $X_j$, and by $\comp(X)$ the total number of components in the allocation X, i.e:
\begin{align*}
\comp(X) = \sum_{j=1}^k \comp(X_j)
\end{align*}
Ideally, we would like that each piece be connected, i.e, $\forall i: \comp(X_i)=1$ and $\comp(X)=k$. This requirement is especially meaningful when the divided resource is a time-interval or a land-resource (e.g. a river-bank), since a contiguous piece of time or land is much easier to use than a collection of disconnected patches. 

However, we will show that a fair division with connected pieces is not always possible.%
\footnote{
This impossibility
appears not only in our one-dimensional theoretic model but also in practical, two-dimensional land division situations. A striking example was the India-Bangladesh border. According to Wikipedia page \emph{India\textendash Bangladesh enclaves}, up to 2015, ``Within the main body of Bangladesh were 102 enclaves of Indian territory, which in turn contained 21 Bangladeshi counter-enclaves, one of which contained an Indian counter-counter-enclave... within the Indian mainland were 71 Bangladeshi enclaves, containing 3 Indian counter-enclaves''. Another example is \emph{Baarle-Hertog} --- a Belgian municipality made of 24 separate parcels of land, most of which are exclaves in the Netherlands.
For more details and examples see the Wikipedia page \emph{List of enclaves and exclaves}. We are grateful to Ian Turton for the references.
}
In case a division with connected pieces is not possible, it is still desirable that the number of connectivity components --- $\comp(X)$ --- be as small as possible. 
When dividing an interval, the components are sub-intervals and their number is one plus the number of \emph{cuts}. Hence, the number of components is minimized by minimizing the number of cuts \citep{Robertson1995Approximating,Webb1997How,Shishido1999MarkChooseCut,Barbanel2004Cake,Barbanel2014TwoPerson}.
In a realistic, 3-dimensional world, the additional dimensions can be used to connect the components, e.g, by bridges or tunnels. Still, it is desirable to minimize the number of components in the original division in order to reduce the number of required bridges/tunnels.%
\footnote{
The goal of minimizing the number of components is pursued not only in cake-cutting papers but also in real-life politics. Going back to India and Bangladesh, after many years of negotiations they finally started to exchange most of their enclaves during the years 2015-2016. This reduced the number of components from 200 to a more reasonable number.
}

\subsection{Fairness criteria}
\label{sub:fairness}
We first define the \emph{family-valuation} functions:
\begin{align*}
\wavg_j(X_j)=\frac{\sum_{i\in F_j}V_i(X_j)}{n_j}\,\,\,\,\,\,\,\,\,\,\,\,\,\textrm{{for}\,}\,j\in\{1,...,k\}.
\end{align*}
Now, an allocation $X$ is called:
\begin{align*}
\text{\apr{}}
&&
\text{if~}
\forall j\range{1}{k}: 
&&
\wavg_j(X_j)\geq t_j;
\\
\text{\upr{}}
&&
\text{if~}
\forall j\range{1}{k}: 
&&
\forall i\in F_j: V_i(X_j)\geq t_j;
\\
\text{$h$-\mpr{}}
&&
\text{if~}
\forall j\range{1}{k},
\end{align*}
\vskip -8mm
\begin{align*}
\text{for at least a fraction $h$ of the members }i\in F_j: 
V_i(X_j)\geq t_j.
\end{align*}
A property of an allocation is called \emph{feasible} if for every $k$ families and $n$ agents there exists an allocation satisfying this property. Otherwise, the property is called \emph{infeasible}. In the following sections we will study the feasibility of the above fairness criteria.

Note that \uprness{} obviously implies both \aprness{} and $h$-\mprness{} for any $h\in[0,1]$. The other two do not imply each other, as shown in the following example.
\begin{example}
Consider an interval consisting of four sub-intervals. It has to be divided between two families: (1) \{Alice,Bob,Chana\} and (2) \{David,Esther,Frank\}. The families have equal entitlements, i.e, $t_1=t_2=1/2$.  Each member's valuation of each sub-interval is shown in the table below:

\begin{center}
\begin{tabular}{|c|c|c|c|c|c|c|}
\hline Alice   & 60 & 30 & 3 & 3 \\
\hline Bob     & 50 & 40 & 3 & 3 \\
\hline Chana & 10 & 80 & 3 & 3 \\
\hline
\hline David   & 3 & 3 & 60 & 30 \\
\hline Esther     & 3 & 3 & 60 & 30 \\
\hline Frank & 3 & 3 & 0  & 90 \\
\hline
\end{tabular}
\end{center}
Note that the value of the entire interval is 96 for all agents. Therefore, \prness{} implies that each family should get a value of at least $48$.

If the two leftmost subintervals are given to family 1 and the two rightmost subintervals are given to family 2, then the division is \emph{\upr{}}, since each member of each family feels that his family's share is worth 90. This division is also, of course, \apr{} and \mpr{}.

If only the single leftmost subinterval is given to family 1 and the other three are given to family 2, then the division is still \emph{${1\over 2}$-\mpr{}}, since Alice and Bob feel that their family received more than 48. However, Chana feels that her family received only 10, so the division is not \upr{}. Moreover, the division is not \apr{} since the average valuation of family 1 is only (60+50+10)/3=40.

If the three leftmost subintervals are given to family 1 and only the rightmost one is given to family 2, then the division is \emph{\apr{}}, since family 2's average valuation of its share is (30+30+90)/3=50. However, it is not \upr{} nor even ${1\over 2}$-\mpr{}, since David and Esther feel that their share is worth only 30. \qed
\end{example}

\section{Average fairness}
\label{sec:avg-fairness}
With average fairness, the family cake-cutting problem can be reduced to the classic problem of cake-cutting among individuals. This gives the following results.
\begin{theorem}
\label{pos-apr}
(a) When families have equal entitlements, \aprness{} with connected pieces (and $k$ components) is feasible.

(b) When families have different entitlements, \aprness{} with connected pieces is infeasible. Moreover, at least $2 k -1$ components may be required for an \apr{} allocation.

(c) When families have different entitlements, \aprness{} with at most $O(k \log{k})$ components
is feasible.
\end{theorem}
\begin{proof}
The positive results --- parts (a) and (c) --- are based on the following reduction.
For each family $F_j$, define a \emph{representative agent} $A_j$ whose valuation is the function $\wavg_j$ defined in Subsection \ref{sub:fairness} above. 
Note that, since the $V_i$ are all nonatomic measures, the $k$ family-valuations $\wavg_j$ are nonatomic measures too. 
By classic results \citep{Steinhaus1948Problem,Even1984Note},
when there are $k$ agents with equal entitlements, there always exists a connected \pr{} division.
As shown in a recent technical report
\citep{SegalHalevi2018Entitlements}, when there are $k$ agents with different entitlements, there always exists a \pr{} allocation with at most 
$2 k \log_2{\widehat{k}} - 2 \widehat{k} + 2$ cuts, where $\widehat{k} := 2^{\lceil\log_2{k}\rceil}=k$ rounded up to the nearest power of two. 
These cuts 
create $2 k \log_2{\widehat{k}} - 2 \widehat{k} + 3$ components.
By definition, such a division is an \apr{} division among the families.

The negative result (b) follows immediately from an identical negative result for individual agents \citep{SegalHalevi2018Entitlements}, by considering $k$ one-member families.
\qed
\end{proof}

\begin{remark}
\label{rem:computation}
Fairness for individuals and average-fairness for families are equivalent only from an existential perspective; from a computational perspective they are quite different. 
A connected \pr{} division among $k$ individual agents with equal entitlements can be found by asking the agents $O(k\log k)$ queries \citep{Even1984Note}.%
\footnote{
In the standard computational model of fair cake-cutting, the 
central planner is allowed to ask each individual agent one of two types of queries: an \emph{eval query} asks about the agent's value for a given piece of cake, and a \emph{mark query} asks the agent to denote a piece with a given value. See \citet{Robertson1998CakeCutting} for a formal definition.
}
However, an \apr{} division cannot be found using a finite number of such queries, even when there are $k=2$ families with two agents in each family, and even without any restrictions on the number of components.
\begin{proof}
The proof is by a reduction from the problem known as \emph{equitable cake-cutting}. This is a variant of fair division among individual agents, in which the fairness criterion for an allocation $X$ is that every two agents have the same subjective utility, i.e. $\forall i,j: V_i(X_i) = V_j(X_j)$.
\citet{Procaccia2017Lower}, extending a previous result by \citet{Cechlarova2012Computability}, showed that an equitable allocation cannot be computed by a finite number of  queries, even when there are only two individual agents and no connectivity constraints.
We show that the same is true for \apr{} allocation.

Given an instance of equitable cake-cutting with two agents with value-measures $V_1$ and $V_2$, construct an instance of \apr{} division with two families, where in each family there are two members with value-measures $V_1$ and $V_2$. We claim that an allocation is equitable in the original problem if-and-only-if it is \apr{} in the new problem:
\begin{align*}
&
\text{The allocation $(X_1,X_2)$ is equitable among the individuals}
\\
\iff &
V_1(X_1) = V_2(X_2)
\\
\iff &
V_1(X_1) = 1 - V_2(X_1) ~~~~ \text{(since $V_2(X_1)+V_2(X_2)=V_2(C)=1$)}
\\
\iff &
[V_1(X_1) + V_2(X_1)]/2 = 1/2 
\\
\iff &
[V_1(X_1) + V_2(X_1)]/2 \geq 1/2 
\text{~~and~~}
[V_1(X_1) + V_2(X_1)]/2 \leq 1/2 
\\
\iff &
[V_1(X_1) + V_2(X_1)]/2 \geq 1/2 
\text{~~and~~}
[V_1(X_2) + V_2(X_2)]/2 \geq 1/2 
\\
&
~~~~~~~~\text{(since $V_i(X_1)+V_i(X_2)=V_i(C)=1$)}
\\
\iff &
\text{The allocation $(X_1,X_2)$ is \apr{} among the families.}
\end{align*}
Hence, if we could find an \apr{} division by a finite number of queries, we could also find an equitable division by a finite number of queries --- a contradiction to \citet{Procaccia2017Lower}.
\qed
\end{proof} 
\end{remark}

\section{Unanimous fairness}
\label{sec:unan-fairness}
Before presenting our results, we note that \uprness{}, like \aprness{}, can also be defined using family-valuation functions. Define:
\begin{align*}
\wmin_j(X_j):=\min_{i\in F_j}V_i(X_j) \,\,\,\,\,\,\,\,\,\,\,\,\,\textrm{{for}\,}\,j\in\{1,...,k\}.
\end{align*}
Then, a division is \upr{} if and only if:
\begin{align*}
\forall j: \wmin_j(X_j)\geq t_j
\end{align*}
However, in contrast to the functions $\wavg$ defined in Section \ref{sec:avg-fairness}, the functions $\wmin$ are in general not additive. For example, suppose $C$ is an interval with three subintervals and a family has the following valuations:
\begin{center}
\begin{tabular}{|c|c|c|c|c|c|}
\hline           & $C_1$ & $C_2$ & $C_3$ & $C_1\cup  C_2\cup C_3$ \\
\hline
\hline Alice      & 1 & 1 & 1 & $3 = 1+1+1$ \\
\hline Bob        & 0 & 2 & 1 & $3 = 0+2+1$ \\
\hline Chana    & 0 & 1 & 2 & $3 = 0+1+2$ \\
\hline
\hline $\wmin$ & 0 & 1 & 1 & $3 > 0+1+1$ \\
\hline
\end{tabular}
\end{center}
While the individual valuations are additive, $\wmin$ is not additive (it is not even subadditive). Therefore, the classic results we used in Theorem \ref{pos-apr} are inapplicable here, and different techniques are needed.

\subsection{Exact division}
Initially, we assume that the entitlements are equal, i.e: $t_j=1/k$ for all $j$. We relate \uprness{} to the problem of finding an \emph{exact division}:%
\footnote{The definition uses capital $N$ and $K$ to distinguish the parameters of exact division from the parameters of unanimous-fair division.}.
\begin{definition}
\exact{}$(N,K)$ is the following problem. Given $N$ agents and an integer $K$, divide $C$ into $K$ pieces, such that each of the $N$ agents assigns exactly the same value to all pieces: 
\begin{align*}
\forall j=1,...,K:\,\,\forall i=1,...,N:\,\,V_i(X_j)=1/K.
\end{align*}
\end{definition}
From an economic perspective, there is little intrinsic value in the concept of exact division.
However, in this section we will prove that it is closely linked to the concept of unanimously-fair division. In fact, we will prove that the existence of a solution to each of these problems implies a solution to the other problem.

Below, we denote by \unpr{}$(n,k)$ the problem of finding a \upr{} division when there are $n$ agents grouped in $k$ families with equal entitlements. 

\subsection{\unpr$\implies$ Exact}
\begin{lemma}
\label{lemma:unprop-implies-exact}
For every pair of integers $N\geq 1,K\geq 1$, a solution to \unpr{} $(N(K-1)+1,\,K)$ implies a solution to \exact{} $(N,K)$
with the same number of components.
\end{lemma}
\begin{proof}
Given an instance of \exact{}$(N,K)$ ($N$ agents and a number $K$ of required pieces), create $K$ families. Each of the first $K-1$ families contains $N$ agents with the same valuations as the given agents. The $K$-th family contains a single agent whose valuation is $V^*$, defined as the average of $V_1,\ldots,V_N$:

\begin{align*}
V^{*}=\frac{1}{N}\sum_{i=1}^{N}V_{i}.
\end{align*}
The total number of agents in all $K$ families is $N(K-1)+1$. Use
\unpr{} $(N(K-1)+1,\,K)$ to find a \upr{} division, $X$. By definition of unanimous fairness, for each agent $i$ in family $j$: $V_i(X_j) \geq 1/K$. 

By construction, each of the first $K-1$ families has an agent with valuation $V_i$. Hence, all $N$ agents value each of the first $K-1$ pieces as at least $1/K$ and:

\[
\forall i=1,...,N:\,\,\,\,\,\,\sum_{j=1}^{K-1}V_{i}(X_{j})\geq\frac{K-1}{K}.
\]
Hence, by additivity, every agent values the $K$-th piece as at most $1/K$:

\[
\forall i=1,...,N:\,\,\,\,\,\,V_{i}(X_{K})\leq1/K.
\]
The piece $X_{K}$ is given to the agent with value measure $V^{*}$, so by \prness{}: $V^{*}(X_{K})\geq1/K$. By construction, $V^{*}(X_{K})$ is the average of the $V_{i}(X_{K})$. Hence:

\[
\forall i=1,...,N:\,\,\,\,\,\,V_{i}(X_{K})=1/K.
\]
Again by additivity:

\[
\forall i=1,...,N:\,\,\,\,\,\,\sum_{j=1}^{K-1}V_{i}(X_{j})=\frac{K-1}{K}.
\]
Hence, necessarily: 

\[
\forall i=1,...,N,\,\,\,\,\,\forall j=1,...,K-1:\,\,\,\,\,\,V_{i}(X_{j})=1/K.
\]
So we have found an exact division and solved \exact$(N,K)$ as required.
\qed
\end{proof}
\citet{Alon1987Splitting} proved that for every $N$ and $K$, an \exact$(N,K)$ division might require at least  $N(K-1)+1$ components. Combining this result with the above lemma implies the following negative result:
\begin{theorem}
\label{neg-upr}
For every $N,K$, let $n=N(K-1)+1$. A \upr{} division for $n$ agents grouped into $K$ families might require at least $n$ components.
\end{theorem}
In particular, \uprness{} with connected pieces is infeasible.

\subsection{Exact $\implies$ \unpr}
\begin{lemma} \label{lemma:exact-implies-unprop}
For every pair of integers $n\geq 2, k\geq 1$, a solution to \exact{} $(n-1,k)$ implies a solution to \unpr{} $(n,k)$ for any grouping of the $n$ agents to $k$ families.
\end{lemma}
\begin{proof}
Suppose we are given an instance of \unpr$(n,k)$, i.e, we are given some $n$ agents grouped into $k$ families. Select $n-1$ agents arbitrarily. Use \exact$(n-1,k)$ to find a partition of $C$ into $k$ pieces, such that each of the $n-1$ agents values each of these pieces at exactly $1/k$. Ask the $n$-th agent to choose a piece with a maximal value for him/her. The average value of a piece is $1/k$, so the piece of maximal value is worth for the $n$-th agent at least $1/k$. Give that piece to the family of the $n$-th agent. Give the other $k-1$ pieces arbitrarily to the remaining $k-1$ families. The resulting division is \upr{}.
\qed
\end{proof}
\citet{Alon1987Splitting} proved that for every $N$ and $K$, \exact$(N,K)$ has a solution with at most $N(K-1)+1$ components (at most $N(K-1)$ cuts). Combining this result with the above lemma implies the following positive result:
\begin{theorem}
\label{pos-upr}
Given $n$ agents in $k$ families with equal entitlements, a \upr{} division with $(n-1)\cdot(k-1)+1$ components is feasible.
\end{theorem}
For $k=2$ families, the number of components in Theorem \ref{pos-upr} is $n$, which matches the lower bound of Theorem \ref{neg-upr}.
For $k>2$ families, the number of components can be made smaller, as explained below.

\subsection{Less components: equal entitlements}
\label{sub:unprop-equal}
The purpose of this subsection is to find a \upr{} allocation with fewer components than the guarantee of Theorem \ref{pos-upr}, when all families have equal entitlements.

We start with an example. Assume there are $k=4$ families. As in Theorem \ref{pos-upr}, using $3(n-1)$ cuts, $C$ can be divided into 4 subsets which are considered equal by $n-1$ members. But for a \upr{} division, it is not required that all members think that all pieces are equal, it is only required that all members believe that their family's share is worth at least $1/4$. This can be achieved as follows:
\begin{itemize}
\item Divide $C$ into two subsets which all $n$ agents value at exactly $1/2$. This is equivalent to solving \exact$(n,2)$, which by \citet{Alon1987Splitting}, can be done with at most $n$ cuts. Call the two resulting subsets West and East.
\item Assign arbitrary two families to West and the other two families to East. Mark by $n_W$ the total number of members in the families assigned to West and by $n_E$ the total number of members assigned to East.
\item Divide the West into two pieces which all $n_W$ agents value at exactly $1/4$; this can be done with $n_W$ cuts. Give a piece to each family. Divide the East similarly using $n_E$ cuts.
\end{itemize}
The first step requires $n$ cuts and the second step requires $n_{W}+n_{E}=n$ cuts too. Hence the total number of cuts required is only $2n$, rather than $3n-1$. 

In fact, two cuts can be saved in each step by excluding two members (from two different families) from the exact division. These members will not think that the division is equal, but they will be allowed to choose their favorite piece for their family. Thus only $2(n-2)$ cuts are required. A simple inductive argument shows that whenever $k$ is a power of 2, $(\log_{2}k)\cdot(n-k/2)$ cuts are required.

When $k$ is not a power of 2, a result by \cite{Stromquist1985Sets} can be used. They prove that, for every fraction $r\in[0,1]$, it is possible to cut a piece of $C$ such that all $n$ agents agree that its value is exactly $r$ using at most $2n-2$ cuts.\footnote{They prove that, if $C$ is a circle, the number of connected components is $n-1$. Hence, the number of cuts is $2n-2$. This is also true when $C$ is an interval, although the number of connected components in this case is $n$.} This can be used as follows:
\begin{itemize}
\item Select integers $l_1,\,l_2\in\{1,...,k-1\}$ such that $l_1+l_2=k$.
\item Apply \cite{Stromquist1985Sets} with $r=l_1/k$: using $2n-4$ cuts, cut a piece $X_1$ that $n-1$ agents value at exactly $l_1/k$. This means that these $n-1$ agents value the other piece, $X_2$, at exactly $l_2/k$.
\item Let the $n$-th agent choose a piece for his family; assign the other families arbitrarily such that $l_1$ families are assigned to piece $X_1$ and the other $l_2$ families to piece $X_2$.
\item Recursively divide piece $X_1$ to its $l_1$ families and piece $X_2$ to its $l_2$ families. 
\end{itemize}
After a finite number of recursion steps, the number of families assigned to each piece becomes 1 and the procedure ends. The number of cuts in each level of the recursion is at most $2n-4$. The depth of recursion can be bounded by $\lceil\log_{2}k\rceil$ by dividing $k$ to halves (if it is even) or to almost-halves (if it is odd; i.e. take $l_1 = (k-1)/2$ and $l_2=(k+1)/2$). Hence:
\begin{theorem}
\label{pos-upr-log}
Given $n$ agents in $k$ families with equal entitlements, a \upr{} division with $\logbound\cdot(2n-4)+1$ components is feasible.
\end{theorem}
Note that Theorem \ref{pos-upr} and Theorem \ref{pos-upr-log} both give upper bounds on the number of components required for \uprness{}. The bound of Theorem \ref{pos-upr} is stronger when $k$ is small and the bound of Theorem \ref{pos-upr-log} is stronger when $k$ is large.

\subsection{Less components: different entitlements}
\label{sub:unprop-diff}
The purpose of this subsection is to find a \upr{} allocation with fewer components than the guarantee of Theorem \ref{pos-upr}, when families may have different entitlements.

When families have different entitlements, the procedure of the previous subsection cannot be used. We cannot let the $n$-th agent select a piece for his family, since the pieces are different. For example, suppose there are two families with entitlements $t_1=1/3,t_2=2/3$. We can divide $C$ into two pieces $X_1,X_2$ such that $n-1$ agents value $X_1$ as 1/3 and $X_2$ as 2/3. So all of them agree that $X_1$ should be given to family 1 and $X_2$ should be given to family 2. But, the $n$-th agent might select the wrong piece for his family. Therefore, the procedure should be modified as follows.

\begin{itemize}
\item Select an integer $l\in\{1,...,k-1\}$.
\item Divide the families into two subsets: $F_1,\dots,F_l$ and $F_{l+1},\dots,F_k$.
\item Apply \cite{Stromquist1985Sets} with $r=\sum_{j=1}^l t_j$: using $2n-2$ cuts, cut a piece $X_1$ which all $n$ agents value at exactly $\sum_{j=1}^l t_j$. This means that all $n$ agents value the other piece, $X_2$, at exactly $\sum_{j=l+1}^k t_j$.
\item Recursively divide piece $X_1$ to $F_1,\dots,F_l$ and piece $X_2$ to $F_{l+1},\dots,F_k$.
\end{itemize}
Here, the number of cuts in each level of the recursion is at most $(2n-2)$. The depth of recursion can be bounded by $\lceil\log_{2}k\rceil$ by choosing $l=k/2$ (if $k$ is even) or $l=(k-1)/2$ (if $k$ is odd). Hence:
\begin{theorem}
\label{pos-upr-different}
Given $n$ agents in $k$ families with different entitlements, a \upr{} division with $\logbound\cdot(2n-2)+1$ components is feasible.
\end{theorem}

To conclude the analysis of \uprness{}, recall that, even for $k=2$ families, \uprness{} is as difficult as exact division and might require the same number of components --- $n$. In the worst case, we might need to give a disjoint component to each member, which negates the concept of division to families. Therefore we now turn to the analysis of an alternative fairness criterion that yields more useful results.

\section{Democratic fairness}
\label{sec:dem-fairness}
Like \uprness{} (Section \ref{sec:unan-fairness}), $h$-\mprness{} too can be defined using family-valuation functions. For example, for $h={1\over 2}$:
\begin{align*}
\wmed_j(X_j):=\frac{\text{median}_{i\in F_j}V_i(X_j)}{n_j}\,\,\,\,\,\,\,\,\,\,\,\,\,\textrm{{for}\,}\,j\in\{1,...,k\}.
\end{align*}
A division is ${1\over 2}$-\mpr{} if and only if:
\begin{align*}
\forall j: \wmed_j(X_j)\geq t_j
\end{align*}
However, the $\wmed$ functions are not additive,\footnote{See the example in the beginning of Section \ref{sec:unan-fairness}. In that example $\wmed$ is identical to $\wmin$.} so again the classic results referred to in Theorem \ref{pos-apr} are inapplicable.

\subsection{A division procedure}
We start with a positive result for families with equal entitlements, which shows that \mprness{} is substantially easier than \uprness{}.

\begin{theorem}
\label{pos-mpr-connected}
For every integer $k\geq 2$, when there are $k$ families with equal entitlements, 
${1\over k}$-\mprness{} with connected pieces is feasible and can be found by an efficient protocol.
\end{theorem}
\begin{proof}
The Dubins-Spanier moving-knife protocol \citep{Dubins1961How} can be adapted to families as follows.
A knife moves continuously over the cake from left to right. Whenever in a certain family at least $1/k$ of its members believe that the cake to the left of the knife is worth at least $1/k$, they shout ``stop'', the cake is cut at the knife location, and the shouting family receives the cake to its left. The division so far is \pr{} for $1/k$ of the members in this family.

In the remaining $k-1$ families, at least $(k-1)/k$ of the members believe the remaining cake is worth at least $(k-1)/k$ of its original value. Dividing the remaining cake recursively using the same procedure yields a division that $1/(k-1)$ of $(k-1)/k$ of the members in each remaining family value as at least $1/(k-1)$ of $(k-1)/k$ of the original value; in other words, the division is \pr{} for at least $1/k$ of the members in each family.
\qed
\end{proof}
Theorem \ref{pos-mpr-connected} is particularly useful for $k=2$ families. It implies the existence of a connected division that is considered fair by at least a weak majority in each family.

Unfortunately, this positive result cannot be improved --- it is impossible to guarantee the support of a weak majority when there are three or more families, and it is impossible to guarantee the support of larger majority when there are two families. This is proved in the following subsection.

\subsection{Three or more families: an impossibility result}
\label{sub:majprop-equal}
This subsection presents a lower bound on the number of components required for a democratic-fair division. 
The lower bound holds not only for \prness{} but even for a much weaker fairness notion called \emph{positivity}.

Given a specific division of $C$ among families, define a \emph{zero} agent as an agent who values his family's share as 0 and a \emph{positive} agent as an agent who believes his family received a share with a positive value. Note that \prness{} implies positivity but not vice-versa.
The following lower bound holds even for positivity, hence it also holds for \prness{}.

\begin{lemma}
\label{lem:neg-mpr}
Assume there are $n=mk$ agents, divided into $k$ families with $m$ members in each family. To guarantee that at least $q$ members in each family are positive, the total number of components may need to be at least: 
\begin{align*}
k\cdot\frac{kq-m}{k-1}
\end{align*}
\end{lemma}
\begin{proof}
Number the families by $j=0,...,k-1$ and the members in each family by $i=0,...,m-1$. Assume that $C$ is the interval $[0,mk]$. In each family $j$, each member $i$ wants only the following interval: $(ik+j,\,ik+j+1)$. Thus there is no overlap between desired intervals of different members. The table below illustrates the construction for $k=2,\,m=3$. The families are \{Alice,Bob,Chana\} and \{David,Esther,Frank\}:
\begin{center}
\begin{tabular}{|c|c|c|c|c|c|c|}
\hline Alice   & 1 & 0 & 0 & 0 & 0 & 0 \\
\hline Bob     & 0 & 0 & 1 & 0 & 0 & 0 \\
\hline Chana & 0 & 0 & 0 & 0 & 1 & 0 \\
\hline
\hline David   & 0 & 1 & 0 & 0 & 0 & 0 \\
\hline Esther     & 0 & 0 & 0 & 1 & 0 & 0 \\
\hline Frank & 0 & 0 & 0 & 0 & 0 & 1 \\
\hline
\end{tabular}
\end{center}

Suppose the piece $X_j$ (the piece given to family $j$) is made of $l\geq1$  components. We can make $l$ members of $F_j$ positive using $l$ intervals of positive length inside their desired areas. However, if $q>l$, we also have to make the remaining $q-l$ members positive. For this, we have to extend $q-l$ intervals to length $k$. Each such extension totally covers the desired area of one member in each of the other families. Overall, each family forces $q-l$ zero members in each of the other families. The number of zero members in each family is thus $(k-1)(q-l)$. 
Adding the $q$ members who must be positive in each family gives the necessary condition: $(k-1)(q-l)+q\leq m$. This
is equivalent to:
\begin{align*}
&
m + l(k-1) \geq k q
\\
\implies &
l \geq {k q - m \over k-1}
\end{align*}
The total number of components is $k\cdot l$, which is the claimed expression.
\qed
\end{proof}
By setting $q = h m$ in Lemma \ref{lem:neg-mpr}, we get the following lower bound on the number of cuts in an $h$-\mpr{} division:
\begin{theorem}
\label{neg-mpr}
For any $h\in[0,1]$, in an $h$-\mpr{} division with $n$ agents grouped into $k$ families, the number of components may need to be at least
\begin{align*}
n\cdot\frac{h k-1}{k-1}
\end{align*}
\end{theorem}
In a \upr{} division $h=1$, so the number of components is at least $n$, which coincides with the lower bound of Theorem \ref{neg-upr}. 
On the other hand, when $h = 1/k$ 
the lower bound is $0$, and indeed we already saw that in this case a connected allocation is feasible (Theorem \ref{pos-mpr-connected}).
However, when $h > 1/k$, 
for sufficiently large $n$, the expression in 
Theorem \ref{neg-mpr} is larger than $k$, which implies that a connected division might not exist. In particular, we cannot guarantee a connected ${1\over 2}$-\mpr{} division for three or more families, and we cannot guarantee a connected $h$-\mpr{} division even for two families if $h>1/2$.

\subsection{Three or more families: positive results}
Suppose we do want a ${1\over 2}$-\mpr{} division for three or more families. How many components are sufficient?

As a first positive result, we can use Theorem \ref{pos-upr-different}, substituting $n/2$ instead of $n$: select half of the members in each family arbitrarily, then find a division which is \upr{} for them while ignoring all other members. This leads to:
\begin{theorem}
\label{pos-mpr-different}
Given $n$ agents in $k$ families with different entitlements, ${1\over 2}$\-/\mprness{} with $1+\logbound\cdot(n-2)$ components is feasible.
\end{theorem}

However, for families with equal entitlements we can do much better. 

\begin{theorem}
\label{pos-mpr-equal}
Given $n$ agents in $k$ families with equal entitlements, ${1\over 2}$\-/\mprness{} with at most
\begin{align*}
\min\bigg(
2+(\lceil k/2\rceil-1)\cdot(n/2-2)
&&
,
&&
2+\lceil\log_{2}\lceil k/2\rceil\rceil\cdot(n-8)
\bigg).
\end{align*}
components is feasible.
\end{theorem}
\begin{proof}
The proof is summarized in Algorithm \ref{alg:pos-mpr-equal}.
\end{proof}
\begin{algorithm}
INPUT: 

- $C := $ the unit interval $[0,1]$.

- $n$ additive agents, all of whom value $C$ as 1.

- A grouping of the agents to $k$ families, $F_1,...,F_k$.\\

OUTPUT: 

A ${1\over 2}$-\mpr{} division of $C$ into $k$ pieces.\\

ALGORITHM:

\textbf{Step 1: Halving}

- Each agent $i=1,...,n$ selects an $x_{i}\in[0,1]$ such that $V_{i}([0,x_{i}])=\frac{\lceil k/2\rceil}{k}$
(this means $\frac{1}{2}$ if $k$ is even and $\frac{k+1}{2k}$ if
$k$ is odd). Note: $V_{i}([x_{i},1])=\frac{\lfloor k/2\rfloor}{k}$.

- For each family $j=1,...,k$, find the median of its members' selections:
$M_{j}=\textrm{{median}}_{i\in F_j}x_{i}$.

- Order the families in increasing order of their medians. Find the
median of the family-medians: $M^{*}=M_{\lceil k/2\rceil}$. Cut $C$ at $x=M^{*}$.

\textbf{Step 2: Sub-division}

- Define the \emph{western families} as the $F_j$ with $j=1,...,\lceil k/2\rceil$.
Let $n_{W}$ be the total number of members in these families. Divide
the interval $[0,M^{*}]$ among the western families using \unpr$(n_{W}/2,\,\lceil k/2\rceil)$.

- Similarly, define the \emph{eastern families} as the $F_j$ with
$j=\lceil k/2\rceil+1,...,k$. There are $\lfloor k/2\rfloor$ such
families. Let $n_{E}$ be their total number of members. Divide the
interval $(M^{*},1]$ among the eastern families using \unpr$(n_{E}/2,\,\lfloor k/2\rfloor)$.

\protect\caption{\label{alg:pos-mpr-equal} ${1\over 2}$-\mpr{} division for $k\geq 2$ families.}
\end{algorithm}

The algorithm works in two steps.

\textbf{Step 1: Halving}. For each family, a location $M_{j}$ is calculated such that, if $C$ is cut at $M_{j}$, half the family members value the interval $[0,M_{j}]$ as at least $\frac{\lceil k/2\rceil}{k}$ and the other half value the interval $[M_{j},1]$ as at least $\frac{\lfloor k/2\rfloor}{k}$.
Then, $C$ is cut in $M^{*}$ --- the median of the family medians.
The $\lceil k/2\rceil$ ``western families'' --- for which $M_{j}\leq M^{*}$
--- are assigned to the western interval of $C$ --- $[0,M^{*}]$.
By construction, at least half the members in each of the western families value $[0,M^{*}]$ as at least $\frac{\lceil k/2\rceil}{k}$.
We say that these members are ``happy''. Similarly, the $\lfloor k/2\rfloor$
eastern families --- for which $M_{j}\geq M^{*}$ --- are assigned to the
eastern interval $(M^{*},1]$; at least half the members in each of these families are ``happy'', i.e, value the interval $(M^{*},1]$ as at least $\frac{\lfloor k/2\rfloor}{k}$.

If there are only two families ($k=2$), then we are done: there is exactly one western family and one eastern family ($\lceil k/2\rceil=\lfloor k/2\rfloor=1$
). For each family $j\in\{1,2\}$, at least half the members of each family value their family's share as at least $1/2$. Hence, the allocation of $X_j$ to family $j$ is ${1\over 2}$-\mpr{}. 

If there are more than two families ($k>2$),
an additional step is required.

\textbf{Step 2: Sub-division}. Each of the two sub-intervals should be further divided among the families assigned to it. In each family $F_j$,
at least $n_j/2$ members are happy. So for each $F_j$, select exactly $n_j/2$ members who are happy. Our goal now is to make sure that these agents remain happy. This can be done using a \upr{} allocation, where only $n_j/2$ happy members in each family (hence $n/2$ members overall) are counted. 

The \upr{} allocation guarantees that every western-happy-member believes that his family's share is worth at least $\frac{\lceil k/2\rceil}{k}\cdot\frac{1}{\lceil k/2\rceil}=\frac{1}{k}$.
Similarly, every eastern-happy-member believes that his family's share is worth at least $\frac{\lfloor k/2\rfloor}{k}\cdot\frac{1}{\lfloor k/2\rfloor}=\frac{1}{k}$.
Hence, the resulting division is ${1\over 2}$\-/\mpr{}.

We now calculate the number of components in the resulting division. One cut is required for the halving step. For the \upr{} division of the western interval, the number of required cuts is at most $(\lceil k/2\rceil-1)\cdot(n_{W}/2-1)$ by Theorem \ref{pos-upr},
and at most $\lceil\log_{2}\lceil k/2\rceil\rceil\cdot(n_{W}-4)$
by Theorem \ref{pos-upr-log}. Similarly, for the eastern interval the number of required cuts is at most the minimum of $(\lfloor k/2\rfloor-1)\cdot(n_{E}/2-1)$
and $\lceil\log_{2}\lfloor k/2\rfloor\rceil\cdot(n_{E}-4)$. The total
number of cuts is thus at most $1+(\lceil k/2\rceil-1)\cdot(n/2-2)$ and at most $1+\lceil\log_{2}\lceil k/2\rceil\rceil\cdot(n-8)$. The total number of components is larger by one. 

\subsection{Comparison and Open Questions}
\label{sec:comparison}
Table \ref{tab:summary-pr} compares the three variants of \prness{}, focusing on families with equal entitlements. Recall that $n$ is the total number of agents in all families.

\begin{table}
\begin{center}
\begin{tabular}{|c|c||c|c|c|}
\hline 
\multirow{2}{2.5cm}{\textbf{\small Fairness notion}} &
\multirow{2}{*}{\textbf{\shortstack{$k$}}} &
\multicolumn{2}{c|}{\textbf{\#Connectivity Components}}
&
\textbf{\shortstack{Compatible\\with PE}}
\\
\cline{3-4} 
&  & \textbf{Lower} & \textbf{Upper} 
&
\\
\hline 
\hline 
\apr{} (Sec. \ref{sec:avg-fairness}) & $k$ & $k$ & $k$ (connected)
& Yes
\\
\hline 
\cline{2-5}
\upr{} & 2 & $n$ & $n$
& \multirow{4}{*}{Yes}
\\
\cline{2-4}
& $3$ & $n$ & $2n-1$
&
\\
\cline{2-4}
& $4$ & $n$ & $2n-3$
&
\\
\cline{2-4}
(Sec. \ref{sec:unan-fairness}) & $k$ & $n$ & \shortstack{$\min(1+\logbound\cdot(2n-4),$ \\ $(k-1)\cdot(n-1)+1)$
}
& 
\\
\hline 
\cline{2-4} 
${1\over 2}$-\mpr{} & 2 & 2 & 2 (connected)
& \multirow{4}{*}{Yes}
\\
\cline{2-4} 
& $3$ & $n/4$ & 
$n/2$
&
\\
\cline{2-4} 
& $4$ & $n/3$ & 
$n/2$
&
\\
\cline{2-4} 
(Sec. \ref{sec:dem-fairness}) & $k$ & $n\cdot\frac{k/2-1}{k-1}$ & 
\shortstack{$\min(2+\lceil\log_{2}\lceil k/2\rceil\rceil\cdot(n-8),$ \\ $2+(\lceil k/2\rceil-1)\cdot(n/2-2))$
}
&
\\
\hline 
\end{tabular}
\end{center}
\caption{\label{tab:summary-pr}
Properties of a \pr{} division with different family-fairness notions and different number of families ($k$).  
The rightmost column considers the compatibility of the fairness notion with Pareto-efficiency.
}
\end{table}
The case of $k=2$ families is well-understood. The results for all fairness criteria are tight: by all fairness definitions, we know that a fair division exists with the smallest possible number of connectivity components. 

The case of $k>2$ families opens some questions:
\begin{itemize}
\item Is \uprness{} with $n$ components feasible for all $k$? (particularly, with $k=3$ families, is the number of required components $n$ as in the lower bound, or $2n-1$ as in the upper bound?).
\item Is ${1\over 2}$-\mprness{} with $n\cdot\frac{k/2-1}{k-1}$ components feasible for all $k$?  (particularly, with $k=3$ families, is the number of required components $n/4$ as in the lower bound, or $n/2$ as in the upper bound?).
\end{itemize}
The case of different entitlements is much less understood even for individual agents \citep{SegalHalevi2018Entitlements}, let alone for families.

What fairness notion is the most practical?
The table shows that it depends on the total number of agents ($n$). When $n$ is small (as is common when dividing an estate among heirs), it is reasonable to try to attain a unanimously-fair division. However, when $n$ is large (as is common when dividing disputed lands among states), unanimous fairness quickly becomes impractical, as the number of components might grow linearly with $n$. 
In this case, we must settle for a weaker fairness criterion. When $k=2$, we can find a democratically-fair allocation
that is also connected. When $k>2$, democratic fairness too might be impractical, and we may have to settle for average-fairness.

\section{Pareto-efficiency}
\label{sec:efficiency}
So far, we studied 
the compatibility of fairness criteria with a 
\emph{geometric} requirement --- reducing the number of connectivity components.
In this section we replace the geometric requirement with an \emph{economic} requirement --- {Pareto efficiency}. 
An allocation is called \emph{Pareto-efficient (PE)} if no other allocation is weakly better for all individual agents and strictly better for some individual agents. 
Fortunately, PE is compatible even with the strongest variant of the \prness{} criterion:
\begin{theorem}
\label{pos-pe-upr}
There always exists an allocation that is both PE and \upr{} (hence also \apr{} and \mpr{}), even when families have different entitlements.
\end{theorem}
\begin{proof}
We use a famous theorem
of \citet{Dubins1961How}, which is a special case of a measure-theoretic theorem by \citet{dvoretzky1951relations}. 

For every partition $X$ of $C$ into $k$ pieces, let $M(X)$ be its \emph{value-matrix} --- an $n$-by-$k$ matrix $M$ where  $\forall i\range{1}{n}, \forall j\range{1}{k}: M_{i,j} = V_i(X_j)$.
Let $\mathbb{M}_C$ be the set of all matrices that correspond to such partitions:
\begin{align*}
\mathbb{M}_C :=
\{
M(X)
| \text{$X$ is a partition of $C$ into $k$ pieces}
\}
\end{align*}
Theorem 1 of \citet{Dubins1961How} implies that the set $\mathbb{M}_C$ is compact.

Define a second set of matrices representing the \upr{} condition:
\begin{align*}
\mathbb{M}_{PR} := 
\{
M\text{ is an $n\times k$ matrix}
| 
\forall j\range{1}{k}:
\forall i\in F_j: 
M_{i,j} \geq t_j
\}
\end{align*}
Finally, define $\mathbb{M}_{CPR} := \mathbb{M}_{C} \cap \mathbb{M}_{PR}$. This set represents all value-matrices of allocations of $C$ that are \upr{}.
By Theorem \ref{pos-upr}, $\mathbb{M}_{CPR}$ is non-empty.
Since $\mathbb{M}_{C}$ is compact and $\mathbb{M}_{PR}$ is closed, their intersection $\mathbb{M}_{CPR}$ is compact. 

Define the following function $U: \mathbb{M}_{CPR} \to \mathbb{R}$:
\begin{align*}
U(M) := \prod_{j=1}^k 
\prod_{i\in F_j}
M_{i,j}
\end{align*}
This is a continuous function, so it has a maximum point in $\mathbb{M}_{CPR}$; let's call it $M^*$.
This matrix corresponds to an allocation $X^*$ that maximizes, among all \upr{} allocations, the product of valuations of all agents:
$\prod_{j=1}^k 
\prod_{i\in F_j}
V_i(X_j)$.
This product is strictly increasing with the value of each agent $i\in N$, so the allocation $X^*$ is Pareto-efficient in the set $\mathbb{M}_{CPR}$.
Since every Pareto-improvement of an allocation in $\mathbb{M}_{CPR}$ is also in $\mathbb{M}_{CPR}$, the allocation $X^*$ is also Pareto-efficient in general.
\qed
\end{proof}

\begin{remark}
\label{pe-connected}
So far we considered two pairs of requirements: fairness+connectivity and fairness+efficiency. This raises the natural question of whether connectivity+efficiency are compatible. We provide two answers.

\paragraph{(a)}
There might not exist a connected allocation that is Pareto-efficient \emph{in the set of all allocations}, even without any fairness considerations, and even with only two individual agents (two singleton families). Moreover, the number of components in such allocation might be unbounded.
\begin{proof}
For any integer $M$, suppose $C$ is the interval $[0, 2M]$. Suppose agent 1  assigns a value of $1$ to the intervals $[0,1], [2,3], [4,5]$, etc., and a value of $0$ to the rest of $C$, and agent 2 assigns a value of $1$ to the intervals $[1,2], [3,4], [5,6]$, etc., and a value of $0$ to the rest of $C$.  Then, any Pareto-efficient allocation has $2M$ connected components.
\qed
\end{proof}

On the other hand:

\paragraph{(b)}
There always exists a connected allocation that is Pareto-efficient \emph{in the set of all allocations with at most $d$ components}, where $d\geq 1$ is any fixed integer. Existence is guaranteed for any number of families with any number of agents, and even with a \uprness{} requirement.
\begin{proof}
Suppose w.l.o.g. that $C$ is the interval $[0,1]$.
Each division with at most $d$ components can be represented by a vector $x_1,\ldots,x_d$ where $\forall i\in[d]: x_i\in[0,1]$ and $\sum_{i=1}^d x_i=1$ (where $x_i$ represents the length of the $i$-th component). Hence the set of all such divisions can be represented by the $(d-1)$-dimensional standard simplex in $\mathbb{R}^d$. This set is compact and convex. Hence, similarly to Theorem \ref{pos-pe-upr}, this set contains an allocation that is both PE and \upr{}.
\qed
\end{proof}
\end{remark}

\section{Envy Freeness}
\label{sec:no-envy}
So far, we used \prness{} as our individual fairness criterion. Another criterion that is very common in economics is \emph{\efness{}}.
We study this criterion for families with equal entitlements.

Analogously to the definitions in subsection \ref{sub:fairness}, we call an allocation $X$ ---
\begin{align*}
\text{\aef{}}
&&
\text{if~}
\forall j,j'\range{1}{k}: 
&&
\wavg_j(X_j)\geq \wavg_j(X_{j'});
\\
\text{\uef{}}
&&
\text{if~}
\forall j,j'\range{1}{k}: 
&&
\forall i\in F_j: V_i(X_j)\geq V_i(X_{j'});
\\
\text{$h$-\mef{}}
&&
\text{if~}
\forall j,j'\range{1}{k},
\end{align*}
\vskip -8mm
\begin{align*}
\text{for at least a fraction $h$ of the members }i\in F_j: 
V_i(X_j)\geq V_i(X_{j'}).
\end{align*}
With individual agents, it is well known that \efness{} implies \prness{} (with equal entitlements).
With two individual agents, \efness{} and \prness{} are equivalent.
The same implications are true with families. 
Each variant of \efness{} implies the corresponding variant of \prness{}.%
\footnote{
Suppose an agent $i \in F_j$ thinks the division is \ef{}.
Then $V_i(X_j)$ is equal to $\max_{j'=1}^{k} V_i(X_{j'})$. The maximum is at least as large as the average, so 
$V_i(X_j)$ is at least as large as an average value of a piece, which is $V_i(C)/k$.}
When there are $k=2$ families, each variant of \efness{} is equivalent to the corresponding variant of \prness{}.%
\footnote{
Suppose an agent $i \in F_j$ thinks the division is \pr{}. Then $V_i(X_j) \geq 1/2$. By additivity, for the other family $j'\neq j$, $V_i(X_{j'}) \leq 1/2$.
Hence 
$V_i(X_j) \geq V_i(X_{j'})$, so agent $i$ thinks the division is \ef{}.
}

Most of our results for \prness{} with equal entitlements are also valid for \efness{}.
For \aefness{}, we can use
classic results proving the existence of \ef{} allocations with connected pieces among individual agents \citep{Stromquist1980How,Su1999Rental}.
Applying the same reduction as in Theorem \ref{pos-apr} we get:

\begin{customthm}{\ref{pos-apr}'}
For any $k$ families, \aefness{} with connected pieces is feasible.
\end{customthm}

Since \efness{} implies \prness{}, the negative results are still valid:

\begin{customthm}{\ref{neg-upr}'}
For every $N,K$, let $n=N(K-1)+1$. A \uef{} division for $n$ agents grouped into $K$ families might require at least $n$ components.
\end{customthm}

Some positive results remain valid too. Lemma \ref{lemma:exact-implies-unprop} is based on an exact division. Therefore it holds, with the same proof, even if we replace \upr{} with \uef{}. Therefore:
\begin{customthm}{\ref{pos-upr}'}
Given $n$ agents in $k$ families, a \uef{} division with $(n-1)\cdot(k-1)+1$ components is feasible.
\end{customthm}

However, the recursive-halving procedure of
Theorem \ref{pos-upr-log} cannot be used here. Suppose we divide $C$ into two subsets, West and East, which all $n$ agents value at exactly $1/2$.
Then, we assign arbitrary $k/2$ families to West and the other families to East. We find an exact division of the West among the western families and an exact division of the East among the eastern families. While this division is \pr{}, it is not necessarily \ef{}, since the agents in the west might envy families in the east and vice versa. 
Therefore, 
while 
the number of components required for \upr{} 
division is in $O(n \log k)$,
the best we can currently say about the number of components required for \uef{} is that it is in $O(n k)$.

The positive result of Theorem \ref{pos-mpr-connected} holds for \efness{} too:
\begin{customthm}{\ref{pos-mpr-connected}'}
For every integer $k\geq 2$, there exists a connected division among $k$ families, that is \ef{} for at least $1/k$ of the members in each family.
\end{customthm}
\begin{proof}
\citet{Su1999Rental} presents a procedure (attributed to Simmons) for finding a connected \ef{} division among $k$ individual agents. It is based on presenting various connected partitions to the agents, and asking each agent which of the $k$ pieces is the best.
He proves that there exists a partition in which each agent gives a different answer; that partition corresponds to an envy-free allocation. He also shows a procedure for finding a sequence of partitions that converges (after possibly infinitely many steps) to that envy-free allocation.

The Simmons-Su procedure can be adapted to families in the following way. Whenever a family is asked ``which of the $k$ pieces is better?'', it answers by doing a \emph{plurality voting} among its members. 
Then, in the final division, each family receives a piece that is considered the best by a plurality of its members, which is at least a fraction $1/k$ of its members. Therefore, at least $1/k$ of each family's members feel that the final allocation is envy-free.
\qed
\end{proof}

Similarly, the negative result for $h$-\mprness{} in Theorem \ref{neg-mpr} is valid for $h$-\mefness{} too. 

Theorem \ref{pos-mpr-equal} about ${1\over 2}$-\mprness{} does not hold as-is for ${1\over 2}$-\mefness{}, but it can be adapted by adapting Algorithm \ref{alg:pos-mpr-equal}. Step 1 --- the halving step --- remains the same.
Step 2 --- the subdivision step --- should be modified to use an exact division, as follows:

\begin{quote}
- Using \exact$(n/2,\lceil k/2\rceil)$
, find an exact division of the interval $[0,M^{*}]$ into $\lceil k/2\rceil$ pieces, such that all $n/2$ happy agents find the pieces equal.
Assign these pieces to the \emph{western families} --- the $F_j$ with $j=1,...,\lceil k/2\rceil$.

- Using \exact$(n/2,\lfloor k/2\rfloor)$
, find an exact division of the interval $(M^{*},1]$ into $\lfloor k/2\rfloor$ pieces, such that all $n/2$ happy agents find the pieces equal.
Assign the pieces to the 
\emph{eastern families} --- $F_j$ with
$j=\lceil k/2\rceil+1,...,k$. 
\end{quote}

The halving step requires a single cut.
The two exact divisions require $(n/2)\cdot  (k-2)$ cuts. Therefore the total number of components is $n(k-2)/2+2$:

\begin{customthm}{\ref{pos-mpr-equal}'}
Given $n$ agents in $k$ families with equal entitlements, ${1\over 2}$\-/\mefness{} with at most $n(k-2)/2+2$ components is feasible.
\end{customthm}

Table \ref{tab:summary-ef} summarizes our results for envy-free division and shows some remaining gaps.
\begin{table}
\begin{center}
\begin{tabular}{|>{\centering}p{2.5cm}|c||c|c|c|}

\hline 
\multirow{2}{2.5cm}{\textbf{Fairness notion}} & \multirow{2}{*}{\textbf{$k$}} & \multicolumn{2}{c|}{\textbf{\#Connectivity Components}} 
\\
\cline{3-4} 
&  & \textbf{Lower bound} & \textbf{Upper bound}
&
\textbf{\shortstack{Compatible\\with PE}} 
\\
\hline 
\hline 
\aef{} & Any & $k$ & $k$ (connected)
& Yes
\\
\hline 
\hline 
& 2 & $n$ & $n$ 
& Yes
\\
\cline{2-4}
\uef{} & 3 & $n$ & $2n-1$
& {No}
\\
\cline{2-4}
(Sec. \ref{sec:unan-fairness}) & 4 & $n$ & $3n-2$
& {No}
\\
\cline{2-4}
& Any & $n$ & $(k-1)\cdot(n-1)+1$
& {No}
\\
\hline 
\hline
& 2 & 2 & 2 (connected) 
& Yes
\\
\cline{2-4} 
${1\over 2}$-\mef{} & 3 & $n/4$ & $n/2+2$ 
& ?
\\
\cline{2-4} 
(Sec. \ref{sec:dem-fairness}) & 4 & $n/3$ & $n+1$ 
& ?
\\
\cline{2-4} 
& Any & $n\cdot\frac{k/2-1}{k-1}$ & $n\cdot(k-2)/2+2$ 
& No ($k\geq 5$)
\\
\hline 
\end{tabular}
\end{center}
\caption{\label{tab:summary-ef}
Properties of envy-free division with different family-fairness notions and different number of families ($k$). The rightmost column considers the compatibility of the fairness notion with Pareto-efficiency.
}
\end{table}

We now consider the combination of envy-freeness with Pareto-efficiency. Some of our positive results from Section \ref{sec:efficiency} are still valid:
\begin{customthm}{\ref{pos-pe-upr}'}
(a) With $k=2$ families, there always exists an allocation that is both PE and \uef{} (hence also $h$-\mef{} for any $h$).

(b) With any number of families, 
there always exists an allocation that is PE and \aef{}.
\end{customthm}
\begin{proof}
(a) With $k=2$ families, \efness{} and \prness{} are equivalent, so this follows directly from Theorem \ref{pos-pe-upr}.

(b) We use the same reduction as in Theorem \ref{pos-apr} and the same compactness argument as in Theorem \ref{pos-pe-upr}.
For each family $F_j$, define a representative agent $A_j$ whose valuation is $\wavg_j$. 
There exists an allocation $X^*$ that maximizes the product of valuations of the representatives: $\prod_{j=1}^k \wavg_j(X_j)$.

\citet[Section 5]{segalhalevi2018monotonicity} prove, in the setting of cake-cutting among individuals, that every allocation maximizing the product of values is envy-free. Therefore, in the allocation $X^*$, there is no envy among the representatives. By definition of \aef{}, $X^*$ is an \aef{} allocation among the families.

The product $\prod_{j=1}^k \wavg_j(X_j)$ is 
strictly increasing with the value of each individual agent $i\in N$. Therefore, the allocation $X^*$ maximizing this product is Pareto-efficient. \qed
\end{proof}

Next, consider Remark \ref{pe-connected} regarding connectivity and efficiency.
Part (a) holds regardless of fairness considerations.
Part (b) implies that there always exists a \upr{} allocation that is PE in the set of connected allocations; this positive result is not true when we replace \prness{} with \efness{}, even with singleton families. This is proved by Example 5.1 in \citet{segal2018resource}.

Even without connectivity constraints, Pareto-efficiency is incompatible with \uefness{} and \mefness{}.

The incompatibility between PE and \uefness{} appears even when we take a minimal step forward from the case of two families: there are three families, only one of which is a couple and the other two are singles.
\begin{theorem}
\label{neg-pe-uef}
With three or more families, there might be no allocation that is both PE and \uef{}.
\end{theorem}
\begin{proof}
The proof is based on an example used by \citet{Bade2018PONE} in the context of fair division of homogeneous goods.
$C$ is an interval composed of two sub-intervals $Y$ and $Z$ of length 1. $C$ has to be divided among three families --- a couple and two singles --- with the following valuations:
\begin{center}
\begin{tabular}{|c|c|c|c|c|}
\hline   & Y & Z\\
\hline    \hline
\hline Alice   & 1 & 1\\
\hline George  & 7 & 1\\
\hline    \hline
\hline Bob     & 2 & 1\\
\hline    \hline
\hline Esther   & 5 & 1\\
\hline
\end{tabular}
\end{center}
Suppose that we have a \uef{} allocation of $C$ among the three families. Denote by $Y_{AG},Z_{AG}$ the lengths of $Y,Z$ given to Alice+George, and similarly $Y_{B},Z_{B},Y_{E},Z_{E}$ are the lengths given to Bob and Esther. Then:
\begin{align*}
7 Y_{AG} + Z_{AG} \geq 7 Y_{B} + Z_{B}
&&
\text{(George does not envy Bob)}
\\
2 Y_{B} + Z_{B} \geq 2 Y_{AG} + Z_{AG}
&&
\text{(Bob does not envy George)}
\\
(7-2) Y_{AG} \geq (7-2) Y_B
&&
\text{(from the above inequalities)}
\\
(*)
~~~~~~~~~~~~~~~~~~
Y_{AG} \geq Y_{B}
\\
1 Y_{AG} + Z_{AG} \geq 1 Y_{B} + Z_{B}
&&
\text{(Alice does not envy Bob)}
\\
2 Y_{B} + Z_{B} \geq 2 Y_{AG} + Z_{AG}
&&
\text{(Bob does not envy Alice)}
\\
(1-2) Y_{AG} \geq (1 - 2) Y_{B} 
&&
\text{(from the above  inequalities)}
\\
(**)
~~~~~~~~~~~~~~~~~~
Y_{AG} \leq Y_{B}
\\
(***)
~~~~~~~~~~~~~~~~~~
Y_{AG} = Y_{B}
&&
\text{(from * and **)}
\\
\implies 
Z_{AG} = Z_{B}
&&
\text{(Bob and Alice+George do not envy)}
\end{align*}
We proved that, in any \uef{} allocation, the share given to Alice+George is identical to the share given to Bob (i.e, the same lengths of both subintervals).
The proof does not depend on the exact valuation functions --- it only depends on the fact that $1 < 2 < 7$, i.e, Bob's valuation of $Y$ is strictly between Alice's and George's valuations. Hence, exactly the same proof works for Esther, i.e:
$
Y_{AG} = Y_{E}
$
and
$
Z_{AG} = Z_{E}$.
Therefore, the shares given to all three families are identical.

We now prove that this allocation cannot be PE. We consider three cases.
\begin{itemize}
\item {Case 1:} $Y_{B} = 0$. Then also $Y_E=Y_{AG}=0$ so $Y$ remains unallocated and the allocation is not PE.
\item {Case 2:} $Z_{E} = 0$. Then also $Z_B=Z_{AG}=0$ so $Z$ remains unallocated and the allocation is not PE.
\item {Case 3:} $Y_B$ and $Z_E$ are positive.
Let $\epsilon = \min(Y_{B},Z_{E}/3)$.
Suppose that 
Bob gives $\epsilon$ of his $Y$
to Esther, and gets in exchange  $3 \epsilon$ of her $Z$.
Then, Bob's value increases by  $3\epsilon-2\epsilon$;
Esther's value increases by  $5\epsilon - 3\epsilon$; and the values of Alice and George are unchanged. This means that the original allocation was not Pareto-efficient.
\qed
\end{itemize}
\end{proof}

Weakening \uef{} to \mef{} does not help when there are 5 or more families.

\begin{theorem}
\label{neg-pe-mef}
With five or more families, there might be no allocation that is both PE and \mef{}.
\end{theorem}
\begin{proof}
The proof is based on an extension of the example of Theorem \ref{neg-pe-uef}, where there are five families --- one triplet and four singles --- with the valuations:
\begin{center}
\begin{tabular}{|c|c|c|c|c|}
\hline   & Y & Z\\
\hline    \hline
\hline Alice   & 1/4 & 1\\
\hline Dina   & 1 & 1\\
\hline George  & 4 & 1\\
\hline    \hline
\hline Bob     & 1/2 & 1\\
\hline    \hline
\hline Chana   & 1/3 & 1\\
\hline    \hline
\hline Esther     & 2 & 1\\
\hline    \hline
\hline Frank   & 3 & 1\\
\hline    \hline
\end{tabular}
\end{center}
Suppose that we have a \mef{} allocation of $C$ among the families. By definition of \mef{}, all singles must not feel any envy.
Moreover, in the first family, at least two members must not feel any envy.  There are three options for the identity of these non-envious members.

\textbf{Option A:} Alice and Dina feel no envy. 
We consider Bob and Chana. The value of $Y$ for each of them is strictly between the value of $Y$ for Alice and the value of $Y$ for Dina.
Therefore, similar arguments as in the proof of Theorem \ref{neg-pe-uef} imply that the three allocations of Chana, Bob, and Alice+Dina+George are identical. Now there are three cases:
\begin{itemize}
\item {Case A1:} $Y_C=0$. Then also $Y_B=Y_{ADG}=0$.
Each of Chana, Bob, and Alice+Dina+George receive at most $1/3$ of $Z$, so Dina's value is at most $1/3$.
Since Dina does not envy Esther and Frank, each of them must receive at most $1/3$ of $Y$. This means that at least $1/3$ of $Y$ remains unallocated, so the allocation is not PE.
\item {Case A2:} $Z_B=0$. Then also $Z_C=Z_{ADG}=0$.
Again Dina's value is at most $1/3$,
so Esther and Frank must receive at most $1/3$ of $Z$, so at least $1/3$ of $Z$ remains unallocated, so the allocation is not PE.
\item {Case A3:} $Y_C$ and $Z_B$ are positive.
Then, Bob can give $\epsilon/2.5$ of his $Z$ to Chana in exchange for $\epsilon$ of her $Y$ (for some small $\epsilon>0$) and attain a Pareto improvement, so the original allocation is not PE.
\end{itemize}

\textbf{Option B:} George and Dina feel no envy. 
We consider Esther and Frank. The value of $Y$ for each of them is strictly between the value of $Y$ for George  and the value of $Y$ for Dina. Therefore the three allocations of Esther, Frank, and Alice+Dina+George are identical. The rest of the proof is analogous to Option A.

\textbf{Option C:} Alice and George feel no envy. 
The value of $Y$ for all the singles is strictly between the value of $Y$ for Alice and the value of $Y$ for George. Therefore, the allocations of all five families must be identical. The rest of the proof is analogous to Theorem \ref{neg-pe-uef}.
\qed
\end{proof}

An interesting question that is left open by Theorems \ref{pos-pe-upr}' and \ref{neg-pe-mef} is what happens when there are 3 or 4 families --- does there always exist an allocation that is both PE and \mef{}?

\section{Related Work}
\label{sec:related}
There are numerous papers about fair division in general and fair cake-cutting in particular. We mentioned some of them in the introduction. Here we survey some work that is more closely related to family-based fairness.

\subsection{Dividing other kinds of resources among families}
In a contemporaneous and independent line of work, several authors have studied the problem of fairly dividing \emph{discrete} items among families
\citep{ManurangsiSu17,Suksompong18,suksompong2018resource,kyropoulou2019almost}

They focused on unanimous fairness. They proved that, in many cases, unanimously-fair allocations do not exist. These results complement our impossibility results for unanimous fairness in dividing a \emph{continuous} resource.
After the publication of their work and our working paper, we joined forces to study democratic-fair allocation of discrete goods among families \citep{segalhalevi2018democratic}.

Recently, \citet{ghodsi2018rent} studied fair division of \emph{rooms and rent} among families, using three notions of fairness which they term strong, aggregate and weak. Their strong fairness is our unanimous fairness; their aggregate fairness is our average fairness; their weak fairness means that at least one agent does not envy.

In another recent work, \citet{Bade2018PONE} studied fair and efficient allocation of \emph{homogeneous} divisible resources among families.

\subsection{Group-envy-freeness and on-the-fly coalitions}
\cite{Berliant1992Fair,Husseinov2011Theory,todo2011generalizing} study the concept of \emph{group-envy-freeness}. They assume the standard model of fair division among \emph{individuals} (and not among families). They define a group-envy-free division as a division in which no coalition of individuals can take the pieces allocated to another coalition with the same number of individuals and re-divide the pieces among its members such that all members are weakly better-off. Coalitions in cake-cutting are also studied by \cite{DallAglio2009Cooperation,DallAglio2014Finding}.

In our setting, the families are pre-determined and the agents do not form coalitions on-the-fly. In an alternative model, in which agents \emph{are} allowed to form coalitions based on their preferences, the family-fair-division problem becomes easier. For instance, it is easy to achieve a \upr{} division with connected pieces between two coalitions: ask each agent to mark its median line, find the median of all medians, then divide the agents to two coalitions according to whether their median line is to the left or to the right of the median-of-medians.

\subsection{Fair division with public goods}
In our setting, the piece given to each family is considered a "public good" in this specific family. The existence of fair allocations of homogeneous goods when some of the goods are public has been studied e.g. by \citet{Diamantaras1992Equity,Diamantaras1994Generalization,Diamantaras1996Set,Guth2002NonDiscriminatory}. In these studies, each good is either private (consumed by a single agent) or public (consumed by all agents). In the present paper, each piece is consumed by all agents in a single family --- a situation not captured by existing public-good models.

\subsection{Matching markets}
Besides fair division, family preferences are important in matching markets, too. For example, when matching doctors to hospitals, usually a husband and a wife who are both doctors want to be matched to the same hospital. This issue poses a substantial challenge to stable-matching mechanisms \citep{Klaus2005Stable,Klaus2007Paths,Kojima2013Matching,ashlagi2014stability}.

The idea of satisfying a certain fairness notion for only a certain fraction of the population, rather than unanimously,
has also been studied in the context of matching markets. For example, \citet{ortega2018social} proves that, when two matching markets are merged and a stable matching mechanism is run, it is impossible to attain monotonic improvement for everyone, but it is possible to attain monotonic improvement for at least half the population.

\ifdefined\FULLVERSION

\subsection{Fairness in group decisions}
The notion of fairness between groups has been studied empirically in the context of the well-known \emph{ultimatum game}. In the standard version of this game, an individual agent (the \emph{proposer}) suggests a division of a sum of money to another individual (the \emph{responder}), which can either approve or reject it. In the group version, either the proposer or the responder or both are groups of agents. The groups have to decide together what division to propose and whether to accept a proposed division.

Experiments by \cite{Robert1997Group,Bornstein1998Individual} show that, in general, groups tend to act more rationally by proposing and accepting divisions which are less fair. \cite{Messick1997Ultimatum}
studies the effect of different group-decision rules while \cite{Santos2015Evolutionary} uses a threshold decision rule which is a generalized version of our majority rule (an allocation is accepted if at least $M$ agents in the responder group vote to accept it).

These studies are only tangentially relevant to the present paper, since they deal with a much simpler division problem in which the divided good is \emph{homogeneous} (money) rather than heterogeneous (cake/land).
\fi

\subsection{Non-additive utilities}
\label{sub:nonadditive}
As explained in Sections \ref{sec:unan-fairness} and \ref{sec:dem-fairness}, the difficulty with \uprness{} and \mprness{} is that the associated family-valuation functions are not additive. It is therefore interesting to compare our work to other works on cake-cutting with non-additive valuations.

\cite{Berliant1992Fair,Maccheroni2003How,DallAglio2005Fair} focus on sub-additive, or concave, valuations, in which the sum of the values of the parts is \emph{more} than the value of the whole. These works are not applicable to our setting, because the family-valuations are not necessarily sub-additive --- the sum of values of the parts might be smaller than the value of the whole (see the example in the beginning of Section \ref{sec:unan-fairness}). 

\cite{Sagara2005Equity,DallAglio2009Disputed,Husseinov2013Existence} consider general non-additive value functions. They provide pure 
existence proofs and do not say much about the nature of the resulting divisions (e.g, the number of connectivity components), which we believe is important in practical division applications.

\cite{Su1999Rental} presents a protocol for envy-free division with connected pieces which does not assume additivity of valuations. However, when the valuations are non-additive, there are no guarantees about the value per agent. In particular, with non-additive valuations, the resulting division is not necessarily \pr{}.

\cite{Mirchandani2013Superadditivity} suggests a division protocol for non-additive valuations using non-linear programming. However, the protocol is practical only when the resource to divide is a collection of a small number of \emph{homogeneous} components, where the only thing that matters is what fraction of each component is allocated to each agent. In contrast, in our model the resource is a single \emph{heterogeneous} good.

Finally, \cite{Berliant2004Foundation,Caragiannis2011Towards,SegalHalevi2015EnvyFree} study specific non-additive value functions which are motivated by geometric considerations (location, size and shape). The present paper contributes to this line of work by studying specific non-additive value functions which are motivated by a different consideration: handling the different preferences of family members. A possible future research topic is to find fair division rules that handle these considerations simultaneously, as both of them are important for fair division of land.

\end{document}